\documentclass[11pt]{article}
\usepackage{latexsym,amssymb,amsmath,amsthm,youngtab,graphicx}
\usepackage{youngtab,enumerate}
\usepackage{epsfig}
\usepackage{verbatim}
\usepackage{color}

\newtheoremstyle{custom}
  {3pt}
  {3pt}
  {\slshape}
  {}
  {\bfseries}
  {.}
  { }
   {}
\theoremstyle{custom}
\newtheorem{theorem}{Theorem}[section]
\newtheorem{proposition}[theorem]{Proposition}

\newtheorem{proposition/definition}[theorem]{Proposition/Definition}
\newtheorem{lemma}[theorem]{Lemma}

\newtheorem{conjecture}[theorem]{Conjecture}

\theoremstyle{definition}
\newtheorem{definition}[theorem]{Definition}

\theoremstyle{remark}
\newtheorem{remark}[theorem]{Remark}

\newtheoremstyle{exercise}
  {3pt}
  {6pt}
  {}
  {}
  {\bfseries}
  {:}
  { }
   {}
\theoremstyle{exercise}
\newtheorem{exercise}[theorem]{Exercise}
\newtheoremstyle{exercises}
  {3pt}
  {6pt}
  {}
  {}
  {\bfseries}
  {:}
  {\newline}
   {}

\theoremstyle{exercise}
\newtheorem{exercises}[theorem]{Exercises}

\input epsf
\def\boxit#1{\vbox{\hrule height1pt\hbox{\vrule width1pt\kern3pt
  \vbox{\kern3pt#1\kern3pt}\kern3pt\vrule width1pt}\hrule height1pt}}

\def\trank{\text{rank}}

\def\BC{\mathbb C}\def\BN{\mathbb N}
\def\BR{\mathbb R}

\def\hd{,...,}

\def\11{\mathbf 1}

\def\l{\lambda}
\def\a{\alpha}

\def\o{\omega}

\def\s{\sigma}

\def\op{{\mathord{\,\oplus }\,}}
\def\otc{{\mathord{\otimes\cdots\otimes}\;}}

\def\ra{{\mathord{\;\rightarrow\;}}}

\def\tr{{\rm trace}\;}

\def\op{\oplus}
\def\BZ{\Bbb Z}

\def\op{\oplus}

\def\s{\sigma}

\def\a{\alpha}

\def\l{\lambda}

\def\FS{\mathfrak  S}

\def\BC{\mathbb  C}

\def\tcodim{\text{codim}}
\def\BR{\mathbb  R}

\def\opc{\op\cdots\op}

\def\hd{, \hdots ,}

\def\ra{\rightarrow}

\def\tdet{\operatorname{det}}

\def\tperm{\operatorname{perm}}
\def\ttrace{\operatorname{trace}}

\def\trank{\operatorname{rank}}

\def\be{\begin{equation}}
\def\ene{\end{equation}}

\def\p{{\bold P}}
\def\np{{\bold N\bold P}}

\def\G{\Gamma}

\newcommand{\dc}{\operatorname{dc}}\newcommand{\edc}{\operatorname{edc}}
\newcommand{\rdc}{\operatorname{rdc}}
\newcommand{\immc}{\operatorname{immc}}\newcommand{\himmc}{\operatorname{himmc}}\newcommand{\abpc}{\operatorname{abpc}}
\newcommand{\labpc}{\operatorname{labpc}}
\newcommand{\dlabpc}{\operatorname{dlabpc}}\newcommand{\hmpc}{\operatorname{hmpc}}

 \newcommand{\transp}{\operatorname{transp}}

\def\trank{{\mathrm {rank}}}

\newcommand{\GL}{\operatorname{GL}}

\newcommand{\Mat}{\operatorname{Mat}}

\def\twt{{\rm weight}}

\def\corank{{\mathrm {corank}}}
\newtheorem*{lemma*}{Lemma}
\renewcommand{\det}{\textup{det}}
\newcommand{\per}{\textup{perm}}

\title{ On the complexity of the permanent in various computational models}
\author{Christian Ikenmeyer\thanks{Max Planck Institute for Informatics, Saarland Informatics Campus, Germany}~ and J.M. Landsberg\thanks{Texas A\&M University, Landsberg partially supported by NSF grant DMS-1405348.}}

\begin{document}
\sloppy

\maketitle

\begin{abstract}
 We  answer 
a question in \cite{LRpermdet}, showing the regular determinantal complexity of
the determinant $\tdet_m$  is $O(m^3)$. We answer questions in,  and generalize results of \cite{DBLP:journals/corr/AravindJ15},
showing there is no rank one determinantal expression for $\tperm_m$ or $\tdet_m$ when $m\geq 3$.
Finally we state and prove several  \lq\lq folklore\rq\rq\ results relating different models of computation.
 \end{abstract}

\section{Introduction}

Let $P(y^1\hd y^M)\in S^m\BC^M$ be a homogeneous polynomial of degree $m$ in $M$ variables.
A {\it size $n$ determinantal expression} for $P$ is an expression:
\be\label{detexpr}
P=\tdet_{n}(\Lambda+ \sum_{j=1}^M X^{j}y^j).
\ene
where $X^j,\Lambda$ are $n\times n$ complex matrices.

The {\it determinantal complexity} of $P$, denoted $\dc(P)$, is the smallest $n$ for which
a size $n$ determinantal  expression exists for $P$. Valiant \cite{vali:79-3} proved that for any polynomial $P$, $\dc(P)$ is finite.
 Let $(y^{i,j})$, $1\leq i,j\leq m$, be linear coordinates on the space of $m\times m$ matrices. Let $\tperm_m:=
\sum_{\s\in \FS_m}y^{1,\s(1)}\cdots y^{m,\s(m)}$ where $\FS_m$ is the permutation group on $m$ letters. 

 Valiant's famous algebraic analog of the $\p\neq\np$ conjecture \cite{vali:79-3} is:

\begin{conjecture} [Valiant \cite{vali:79-3}]\label{valperdet} 
The sequence $\dc(\tperm_m)$ grows super-polynomially fast.
\end{conjecture}

The state of the art regarding determinantal expressions for
$\tperm_m$ is   $2^m-1\geq \dc(\tperm_m)\geq \frac{m^2}2$, respectively \cite{Gre11,MR2126826}.  

In the same paper \cite{vali:79-3}, Valiant also made the potentially stronger conjecture that there is
no polynomial sized arithmetic circuit computing $\tperm_m$.

\medskip

There are two approaches towards 
conjectures such as Conjecture \ref{valperdet}.
One is to  first  prove them in {\it restricted models}, i.e., assuming
extra hypotheses, 
with the goal of proving a conjecture by first proving it under weaker and weaker 
supplementary hypotheses until one arrives at the 
original conjecture. The second is to fix a complexity measure
such as $\dc(\tperm_m)$ and then to prove lower bounds
on the complexity measure, which we will call  {\it benchmarks}, and then 
improve the benchmarks. If one takes the first approach, it is important
to be able to compare various restrictions. If one takes the second, and would
like the flexibility of working in different (polynomially) equivalent models, one
needs precise (not just polynomial) relations between the complexity measures.
{\it The primary purpose of this paper is to address these two issues. 

\medskip

We begin with comparing restrictions:}

The first super-polynomial  lower bound for the permanent in any
non-trivial restricted model of computation was proved by Nisan  in \cite{Nisan:1991:LBN:103418.103462}: non-commutative formulas.

 To our knowledge, the first exponential lower bound
for the permanent that does  not also hold for 
the determinant  in any restricted model was $\binom{2m}m -1$ in \cite{LRpermdet}. This model was
{\it equivariant} determinantal expressions (see \cite{LRpermdet} for the definition). Let $\edc(P)$ denote the equivariant determinantal
complexity of $P$. While $\edc(\tdet_m)=m$, and for a generic polynomial $P$, $\edc(P)=\dc(P)$,
in \cite{LRpermdet} it was shown that  $\edc(\tperm_m)=\binom{2m}m -1$.
This paper is a follow-up to \cite{LRpermdet}. While equivariance is natural
for geometry, it is not a typical restriction
in computer science. 

 The   restricted models in this paper have already appeared  in the computer science literature: 
Raz's multi-linear circuits \cite{MR2535881},
Nisan's	non-commutative formulas \cite{Nisan:1991:LBN:103418.103462} and the \lq\lq rank-$k$\rq\rq\ determinantal expressions of
Aravind and Joglekar \cite{DBLP:journals/corr/AravindJ15}.

 Our results regarding different restricted models are:  
\begin{itemize}

\item We answer a question in \cite{LRpermdet} regarding the regular determinantal complexity of the determinant, Proposition \ref{rdcdetm}. 

\item We prove $\tperm_m$ does not admit a rank one determinantal expression for $m\geq 3$, Theorem \ref{rankone}, answering
a question posed in 
\cite{DBLP:journals/corr/AravindJ15}.
\end{itemize}

Regarding benchmarks, we make precise comparisons between different complexity measures, Theorem \ref{cmeasures}.  Most of these relations
were \lq\lq known to the experts\rq\rq\ in terms of the measures being
polynomially related, but for the purposes of comparisons we
need the more precise results presented here.  In particular the     homogeneous iterated matrix
multiplication complexity is polynomially equivalent to determinantal complexity.

 \subsection*{Acknowledgments}
We thank Neeraj Kayal for pointing us towards the himmc model of computation and Michael Forbes for important discussions.  {We also thank Michael Forbes and Amir
Shpilka for  help with the literature and     exposition.}
\section{Definitions, results, and overview}\label{defsect}

We first review the complexity measures corresponding to algebraic branching programs and iterated matrix
multiplication:
  
\begin{definition}[Nisan \cite{Nisan:1991:LBN:103418.103462}]
 An {\it Algebraic Branching Program} (ABP) over $\BC$ is a directed acyclic graph $\Gamma$ with a single
source $s$ and exactly one sink $t$. Each edge $e$ is labeled
with an affine linear function $\ell_e$ in the variables $\{y^i | 1 \leq i \leq M\}$. Every directed path
$p = e_1e_2 \cdots e_k$ represents the product $\Gamma_p := \prod_{j=1}^k
  \ell_{e_j}$ . For each vertex $v$ the polynomial $\Gamma_v$ is defined as $\sum_{p\in \mathcal{P}_{s,v}}\Gamma_p$ 
  where $\mathcal{P}_{s,v}$ is the set of paths from $s$ to $v$.
We say that $\Gamma_v$ \emph{is computed by $\Gamma$ at $v$}.
We also say that $\Gamma_t$ \emph{is computed by $\Gamma$} or that $\Gamma_t$ \emph{is the output of} $\Gamma$.

The {\it size} of $\Gamma$ is the number of vertices. 
Let $\abpc(P)$ denote the smallest size of an algebraic branching program that computes $P$.

An ABP is {\it layered} if we can assign a layer $i\in \BN$ to each vertex such that for all $i$, all edges from layer $i$
go to layer $i+1$. Let $\labpc(P)$ denote the the smallest size of a layered algebraic branching program that computes $P$.
Of course $\labpc(P)\geq \abpc(P)$.

An ABP is {\it homogeneous} if the polynomials computed at each vertex are all homogeneous.

A homogeneous ABP $\Gamma$ is {\it degree layered} if $\Gamma$ is layered and the layer of a vertex $v$ coincides with the degree of~$v$.
For a homogeneous $P$ let $\dlabpc(P)$ denote the the smallest size of a degree layered algebraic branching program that computes $P$.
Of course $\dlabpc(P) \geq \labpc(P)$.
\end{definition}

\begin{definition}
The  {\it iterated matrix multiplication complexity}  of
 a  polynomial $P(y)$ in $M$ variables, $\immc(P)$  is the smallest
$n$ such that there exists affine  linear maps $B_j: \BC^M\ra \Mat_n(\BC)$, $j=1\hd n$, such 
that $P(y)=\ttrace(B_n(y)\cdots B_1(y))$.
The {\it homogeneous iterated matrix multiplication complexity} of 
 a degree $m$ homogeneous polynomial $P\in S^{m}\BC^M$, $\himmc(P)$, is the smallest
$n$ such that there exist natural numbers $n_1\hd n_m$ with
$1=n_1$, and $n=n_1+\cdots + n_m$, and linear maps
$A_s: \BC^M\ra \Mat_{n_s\times n_{s+1}}$, $1\leq s\leq m$,  with the convention $n_{m+1}=n_1$, such that
$P(y)=A_m(y)\cdots A_1(y)$.
\end{definition}

A determinantal  expression \eqref{detexpr} is called {\it regular} if $\trank \Lambda=n-1$.
The {\it regular determinantal complexity} of $P$, denoted $\rdc(P)$, is the smallest $n$
for which a regular size $n$  determinatal expression exists.
 Von zur Gathen \cite{MR910987} showed that any determinantal expression of a polynomial
  whose singular locus has codimension at least five, e.g., the permanent, must be regular. In particular  $\rdc(\tperm_m)=\dc(\tperm_m)$.

 All the interesting regular determinantal expressions for   the permanent
and determinant that we are aware of correspond to homogeneous iterated matrix
multiplication expressions of the exact same complexity.  For example, the
expressions for $\tdet_m$ at the end of \S\ref{abpsect} are iterated
matrix multiplication,  
where  if the block matrices are labeled from left to right $B_1\hd B_m$, the product is
$B_m\cdots B_1$.

In \cite{LRpermdet}, it was shown that if one assumes   that
the symmetry group of the expression captures about half the symmetry
group of $\tperm_m$, then the smallest size such determinantal expression equals the known upper bound of $2^m-1$.  A key to the proof was the
utilization of the Howe-Young duality endofunctor  that exchanges
symmetrization and skew-symmetrization. Indeed, the result was
first proved for  half equivariant  regular determinantal expressions for the determinant, where the proof  was not so difficult, 
and then the endofunctor served as a guide as to how one
would need to prove it for the permanent. 
This motivated Question 2.18 of \cite{LRpermdet}: What is the growth
of the function $\rdc(\tdet_m)$?

\begin{proposition}\label{rdcdetm}  $\rdc(\tdet_m)\leq \frac 13(m^3-m)+1$.
\end{proposition}

Proposition \ref{rdcdetm} is proved in \S\ref{abpsect},
where we show how to translate an ABP for a polynomial $P$  into a regular determinantal
expression for $P$. Translating work of Mahjan-Vinay \cite{MV:97} to determinantal expressions then gives the result.

\medskip

Consider the following variant on multi-linear circuits and formulas: 
Let $M = M_1+\cdots+M_m$ and let
$P\in \BC^{M_1}\otc \BC^{M_m}\subset  S^{m}(\BC^{M_1}\opc\BC^{M_m})$
be a multi-linear polynomial  (sometimes called a {\it set-multilinear polynomial}
in the computer science literature). We say a   homogeneous iterated
matrix multiplication (IMM)  presentation of $P$
is {\it block multi-linear} if each $A_j : \BC^M\ra \Mat_{n_j\times n_{j+1}}$ is
non-zero on exactly one factor.
The size $2^m-1$ determinantal expressions of \cite{Gre11,LRpermdet},
 as mentioned above, translate directly to homogeneous
iterated matrix multiplication expressions.
When one does this translation,  the resulting expressions 
  are block multilinear,
where we assume that the $M=m^2$ variables of $\tperm_m$ or $\det_m$   are  grouped column-wise, so $M_j=m$ for all $1 \leq j \leq m$.
We call block multilinear expressions with this grouping  \emph{column-wise multilinear}. 
 That is, a column-wise multilinear ABP for the determinant is an iterated matrix
multiplication, where each matrix only references variables from a single column
of the original matrix.

 The lower bound in the  following result appeared in \cite{Nisan:1991:LBN:103418.103462} in slightly different language:

\begin{theorem} \label{expthm} The smallest size column-wise multilinear IMM presentation 
of $\tdet_m$ and $\tperm_m$ is $2^m-1$.
 When translated to the regular determinantal expression model,
these expressions respectively correspond  to Grenet's expressions \cite{Gre11} in the case
of the permanent and the expressions  of \cite{LRpermdet} in the case of the determinant.
\end{theorem}

\begin{remark}\label{strength}  The $2^m-1$ lower bound for the permanent (resp. determinant) of \cite{LRpermdet} was obtained by assuming
\lq\lq half-equivariance\rq\rq :
equivariance with respect to left multiplication by diagonal matrices
with determinant one (the torus)
and permutation matrices (resp. equivariance with respect to left multiplication by
matrices with determinant one and assuming a regular expression).   
The optimal   determinantal expression  for the permanent or determinant with half-equivariance
is   equivalent to a
  column-wise multilinear  homogeneous iterated matrix multiplication
expression of the same size, as can be seen in the proofs in \cite{LRpermdet}.  
On the other hand,
column-wise multilinear IMM presentations do not   imply half-equivariance,
  nor is there an   implication in the other direction. It is interesting that
  these two different restricted models have the same optimal expression.  
\end{remark}

Theorem \ref{expthm} even holds \lq\lq locally\rq\rq :

\begin{theorem}\label{nosqueezethm}  
Any IMM presentation of $\tdet_m(y)$ or $\tperm_m(y)$
with   a size   $L\times R$ sub-matrix  of $y$ appearing only in 
$A_1\hd A_{L}$ must have size
at least
$\binom {R} {L}$.
\end{theorem}

\begin{remark} Theorem \ref{nosqueezethm} shows that
if $L(m),R(m)$ are functions such that
$\binom{R(m)}{L(m)}$ grows super-polynomially,
  any sequence  of 
IMM   presentations of $\tperm_m$ (resp.
  IMM presentations  of $\tdet_m$) of polynomial size cannot
have  a size   $L(m)\times R(m)$ sub-matrix
(or a size $R(m)\times L(m)$ sub-matrix)   of $y$ appearing only in 
$A_1\hd A_{L(m)}$  or $A_{m}\hd A_{m-L(m)}$.
In particular, if $R(m)=\a m$ for some constant $0<\a\leq 1$, then
to have a polynomial size presentation, 
$L(m)$ must be bounded above by a constant. 
\end{remark}

\medskip

Our second restricted model  comes from \cite{DBLP:journals/corr/AravindJ15}.
In \cite{DBLP:journals/corr/AravindJ15} they  introduce {\it read-$k$ determinants},  determinantal expressions  where
the $X^{ij}$ have at most $k$ nonzero entries, and show that $\tperm_m$ cannot be expressed
as a read once determinant over $\BR$ when $m\geq 5$. The notion of read-$k$ is not natural from a geometric perspective
as it is not preserved by the group preserving $\tdet_n$, however in section 5 of  the same paper  
they suggest a more meaningful analog inspired by \cite{MR2745772}    called {\it rank}-$k$ determinants:

\begin{definition} A polynomial 
$P(y^1\hd y^M)$ admits a  \it  rank  $k$  determinantal expression if there is a determinantal expression
$P(y)=\tdet(\Lambda+ \sum_jy^jX^j )$ with $\trank X^{j}\leq k$.
\end{definition}

 This definition is   reasonable when $P$ is the permanent because the
individual $y^{i,j}$ are defined up to scale.  In \S\ref{rankksect} we show:

\begin{theorem}\label{rankone}   {Neither   $\tperm_m$ nor $\tdet_m$  admits a rank one  regular determinantal expression over $\BC$ when $m\geq 3$. 
In particular, either   $\tperm_m$ nor $\tdet_m$  admits a read once regular determinantal expression over $\BC$ when $m\geq 3$.  }
\end{theorem}

\begin{remark}  {Anderson, Shpilka and Volk (personal communication
from Shpilka) have shown that if a polynomial $P$ in $n$ variables admits a rank
$k$ determinantal expression of size $s$, then it admits a read-$k$ determinantal
expression
of size $s+2nk$. This combined with the results of \cite{DBLP:journals/corr/AravindJ15}
gives an alternative proof of Theorem \ref{rankone} over $\BR$  
and finite fields where  $-3$ is a quadratic non-residue for $m\geq 5$.}
\end{remark}

\section{Algebraic branching programs and determinants}\label{abpsect}

In this section we describe how to obtain a size $O(m^3)$ regular determinantal expression for  $\tdet_m$.
 We use standard techniques about algebraic branching programs and an algorithm described by Mahajan and Vinay \cite{MV:97}.

\begin{proposition}\label{pro:constructadjmatrix} Let $P$ be a polynomial. Then 
$\dc(P)\leq \labpc(P)-1$.
Moreover, if the constant term  of $P$ is zero, then we also have $\rdc(P)\leq \labpc(P)-1$.
\end{proposition}
\begin{proof}
>From a layered  algebraic branching program $\G^\text{algbp}$ we create a directed graph $\G^\text{root}$ 
by identifying the source and the sink vertex
and by calling the resulting vertex the root vertex.
>From $\G^\text{root}$ we create a directed graph $\G^{\text{loops}}$
by adding at each non-root vertex a loop that is labeled with the constant~1.
Let $A$ denote the adjacency matrix of $\G^{\text{loops}}$.
Since $\G^\text{algbp}$ is layered, each path from the source to the sink in $\G^\text{algbp}$ has the same length.
If that length is even, then $\det(A)$ equals the output of $\G^\text{algbp}$,
otherwise $-\det(A)$ equals the output of $\G^\text{algbp}$.
This proves the first part.

 Now assume $P$ has no constant term. 
Let $\Lambda$ denote the constant part of $A$, so $\Lambda$ is a complex square matrix.
Since $\G^\text{algbp}$ is layered we ignore all edges coming out of the sink vertex of $\G^{\text{algbp}}$
and   order  all vertices of $\G^{\text{algbp}}$ topologically, i.e., if there is an edge from vertex $u$ to
vertex $v$, then $u$ precedes $v$ in the order.
We use this order  to specify   the order in which we write down $\Lambda$.
Since the order is topological, $\Lambda$ is lower triangular with one exception: The first row
can have additional nonzero entries.
By construction of the loops in $\G^{\text{loops}}$ the main diagonal of $\Lambda$ is filled with 1s everywhere but at the top left where $\Lambda$ has a 0.
Thus $\corank(\Lambda)=1$ or $\corank(\Lambda)=0$.
But if $\corank(\Lambda)=0$, then the constant  term of $P$ is $\det(\Lambda)\neq 0$,
which is a contradiction to the assumption.
\end{proof}

\begin{proposition}\label{pro:MahVin}
$\labpc(\det_m) \leq \frac {m^3} 3 - \frac m 3 + 2$.
\end{proposition}
\begin{proof}
This is an analysis of the algorithm in \cite{MV:97}
with all improvements that are described in the article.
We construct an explicit layered ABP $\G$. 
Each vertex of $\G$ is a triple of three nonnegative integers $(h,u,i)$, where $i$ indicates its layer.
The following triples appear as vertices in $\G$.
\begin{itemize}
\item The source $(1,1,0)$.
\item For all $1 \leq i < m$:
 \begin{itemize}
 \item The vertex $(i+1,i+1,i)$.
 \item For each $2 \leq u \leq m$ and each $1 \leq h \leq \min(i,u)$ the vertex $(h,u,i)$.
 \end{itemize}
\item The sink $(1,1,m)$.
\end{itemize}
 
\begin{lemma}
The number of vertices in $\G$ is $\frac {m^3} 3 - \frac m 3 + 2$.
There is only the source vertex in layer~0 and only the sink vertex in layer~$m$.
The number of vertices in layer $i \in \{1,\ldots,m-1\}$ is $i(i+1)/2+i(m-1)$.
\end{lemma}
 
\begin{proof}
By the above construction, the number of vertices in $\G$ equals
\[
2+\sum_{i=1}^{m-1}\Big(1+\sum_{u=2}^m \min(i,u)\Big) = 1+ m + \sum_{i=1}^{m-1}\sum_{u=2}^m \min(i,u).
\]
We see that $\sum_{i=1}^{m-1}\sum_{u=2}^m \min(i,u) = (m-2)(m-1)/2 + \sum_{i=1}^{m-1}\sum_{u=1}^{m-1} \min(i,u)$.
It is easy to see that
$\sum_{i=1}^{m-1}\sum_{u=1}^{m-1} \min(i,u)$ yields the square pyramidal numbers (OEIS\footnote{http://oeis.org/} A000330):
  $m(m-1)(m-\tfrac 1 2)/3$.
Therefore
\[
1+m + \sum_{i=1}^{m-1}\sum_{u=2}^m \min(i,u) = 1+m+ m(m-1)(m-\tfrac 1 2)/3 + (m-2)(m-1)/2 = \tfrac{m^3} 3 - \tfrac m 3 + 2.
\]
To analyze a single layer $1 \leq i \leq m-1$ we observe
\[
1+\sum_{u=2}^m \min(i,u) = \sum_{u=1}^m \min(i,u) = i(i+1)/2 + i(m-i).
\]
\end{proof}

We now describe the edges in $\G$.
The vertex $(h,u,i)$ is positioned in the $i$th layer with only edges to the layer $i+1$,
with the exception that layer $m-1$ has edges only to the sink.
>From $(h,u,i)$ we have the following outgoing edges.
\begin{itemize}
\item If $i+1<m$:
 \begin{itemize}
  \item for all $h+1\leq v \leq m$ an edge to $(h,v,i+1)$ labeled with $x^u_v$.
  \item for all $h+1\leq h' \leq m$ an edge to $(h',h',i+1)$ labeled with $-x^u_h$.
 \end{itemize}
\item If $i+1=m$: An edge to the sink labeled with $\alpha x^u_h$, where $\alpha=1$ if $m$ is odd and $\alpha=-1$ otherwise.
\end{itemize}

The fact that $\G$ actually computes $\tdet_m$ follows from \cite{MV:97}.
\end{proof}

As an illustration for $m=3,4,5$ we  include  the adjacency matrices of the $\G^{\text{loops}}$ that come out of the combination of
the constructions in Proposition~\ref{pro:MahVin} and Proposition~\ref{pro:constructadjmatrix}.

{
\tiny

\begin{verbatim}
   0    0    0    0  x21  x31  x22  x32  x33
 x12    1    0    0    0    0    0    0    0
 x13    0    1    0    0    0    0    0    0
-x11    0    0    1    0    0    0    0    0
   0  x22  x32    0    1    0    0    0    0
   0  x23  x33    0    0    1    0    0    0
   0 -x21 -x31    0    0    0    1    0    0
   0    0    0  x23    0    0    0    1    0
   0 -x21 -x31 -x22    0    0    0    0    1
\end{verbatim}

\begin{verbatim}
   0    0    0    0    0    0    0    0    0    0    0    0 -x21 -x31 -x41 -x22 -x32 -x42 -x33 -x43 -x44
 x12    1    0    0    0    0    0    0    0    0    0    0    0    0    0    0    0    0    0    0    0
 x13    0    1    0    0    0    0    0    0    0    0    0    0    0    0    0    0    0    0    0    0
 x14    0    0    1    0    0    0    0    0    0    0    0    0    0    0    0    0    0    0    0    0
-x11    0    0    0    1    0    0    0    0    0    0    0    0    0    0    0    0    0    0    0    0
   0  x22  x32  x42    0    1    0    0    0    0    0    0    0    0    0    0    0    0    0    0    0
   0  x23  x33  x43    0    0    1    0    0    0    0    0    0    0    0    0    0    0    0    0    0
   0  x24  x34  x44    0    0    0    1    0    0    0    0    0    0    0    0    0    0    0    0    0
   0 -x21 -x31 -x41    0    0    0    0    1    0    0    0    0    0    0    0    0    0    0    0    0
   0    0    0    0  x23    0    0    0    0    1    0    0    0    0    0    0    0    0    0    0    0
   0    0    0    0  x24    0    0    0    0    0    1    0    0    0    0    0    0    0    0    0    0
   0 -x21 -x31 -x41 -x22    0    0    0    0    0    0    1    0    0    0    0    0    0    0    0    0
   0    0    0    0    0  x22  x32  x42    0    0    0    0    1    0    0    0    0    0    0    0    0
   0    0    0    0    0  x23  x33  x43    0    0    0    0    0    1    0    0    0    0    0    0    0
   0    0    0    0    0  x24  x34  x44    0    0    0    0    0    0    1    0    0    0    0    0    0
   0    0    0    0    0 -x21 -x31 -x41    0    0    0    0    0    0    0    1    0    0    0    0    0
   0    0    0    0    0    0    0    0  x23  x33  x43    0    0    0    0    0    1    0    0    0    0
   0    0    0    0    0    0    0    0  x24  x34  x44    0    0    0    0    0    0    1    0    0    0
   0    0    0    0    0 -x21 -x31 -x41 -x22 -x32 -x42    0    0    0    0    0    0    0    1    0    0
   0    0    0    0    0    0    0    0    0    0    0  x34    0    0    0    0    0    0    0    1    0
   0    0    0    0    0 -x21 -x31 -x41 -x22 -x32 -x42 -x33    0    0    0    0    0    0    0    0    1
\end{verbatim}

\newsavebox\myv
\begin{lrbox}{\myv}\begin{minipage}{\textwidth}
\begin{verbatim}
   0    0    0    0    0    0    0    0    0    0    0    0    0    0    0    0    0    0    0    0    0    0    0    0    0    0    0  x21  x31  x41  x51  x22  x32  x42  x52  x33  x43  x53  x44  x54  x55
 x12    1    0    0    0    0    0    0    0    0    0    0    0    0    0    0    0    0    0    0    0    0    0    0    0    0    0    0    0    0    0    0    0    0    0    0    0    0    0    0    0
 x13    0    1    0    0    0    0    0    0    0    0    0    0    0    0    0    0    0    0    0    0    0    0    0    0    0    0    0    0    0    0    0    0    0    0    0    0    0    0    0    0
 x14    0    0    1    0    0    0    0    0    0    0    0    0    0    0    0    0    0    0    0    0    0    0    0    0    0    0    0    0    0    0    0    0    0    0    0    0    0    0    0    0
 x15    0    0    0    1    0    0    0    0    0    0    0    0    0    0    0    0    0    0    0    0    0    0    0    0    0    0    0    0    0    0    0    0    0    0    0    0    0    0    0    0
-x11    0    0    0    0    1    0    0    0    0    0    0    0    0    0    0    0    0    0    0    0    0    0    0    0    0    0    0    0    0    0    0    0    0    0    0    0    0    0    0    0
   0  x22  x32  x42  x52    0    1    0    0    0    0    0    0    0    0    0    0    0    0    0    0    0    0    0    0    0    0    0    0    0    0    0    0    0    0    0    0    0    0    0    0
   0  x23  x33  x43  x53    0    0    1    0    0    0    0    0    0    0    0    0    0    0    0    0    0    0    0    0    0    0    0    0    0    0    0    0    0    0    0    0    0    0    0    0
   0  x24  x34  x44  x54    0    0    0    1    0    0    0    0    0    0    0    0    0    0    0    0    0    0    0    0    0    0    0    0    0    0    0    0    0    0    0    0    0    0    0    0
   0  x25  x35  x45  x55    0    0    0    0    1    0    0    0    0    0    0    0    0    0    0    0    0    0    0    0    0    0    0    0    0    0    0    0    0    0    0    0    0    0    0    0
   0 -x21 -x31 -x41 -x51    0    0    0    0    0    1    0    0    0    0    0    0    0    0    0    0    0    0    0    0    0    0    0    0    0    0    0    0    0    0    0    0    0    0    0    0
   0    0    0    0    0  x23    0    0    0    0    0    1    0    0    0    0    0    0    0    0    0    0    0    0    0    0    0    0    0    0    0    0    0    0    0    0    0    0    0    0    0
   0    0    0    0    0  x24    0    0    0    0    0    0    1    0    0    0    0    0    0    0    0    0    0    0    0    0    0    0    0    0    0    0    0    0    0    0    0    0    0    0    0
   0    0    0    0    0  x25    0    0    0    0    0    0    0    1    0    0    0    0    0    0    0    0    0    0    0    0    0    0    0    0    0    0    0    0    0    0    0    0    0    0    0
   0 -x21 -x31 -x41 -x51 -x22    0    0    0    0    0    0    0    0    1    0    0    0    0    0    0    0    0    0    0    0    0    0    0    0    0    0    0    0    0    0    0    0    0    0    0
   0    0    0    0    0    0  x22  x32  x42  x52    0    0    0    0    0    1    0    0    0    0    0    0    0    0    0    0    0    0    0    0    0    0    0    0    0    0    0    0    0    0    0
   0    0    0    0    0    0  x23  x33  x43  x53    0    0    0    0    0    0    1    0    0    0    0    0    0    0    0    0    0    0    0    0    0    0    0    0    0    0    0    0    0    0    0
   0    0    0    0    0    0  x24  x34  x44  x54    0    0    0    0    0    0    0    1    0    0    0    0    0    0    0    0    0    0    0    0    0    0    0    0    0    0    0    0    0    0    0
   0    0    0    0    0    0  x25  x35  x45  x55    0    0    0    0    0    0    0    0    1    0    0    0    0    0    0    0    0    0    0    0    0    0    0    0    0    0    0    0    0    0    0
   0    0    0    0    0    0 -x21 -x31 -x41 -x51    0    0    0    0    0    0    0    0    0    1    0    0    0    0    0    0    0    0    0    0    0    0    0    0    0    0    0    0    0    0    0
   0    0    0    0    0    0    0    0    0    0  x23  x33  x43  x53    0    0    0    0    0    0    1    0    0    0    0    0    0    0    0    0    0    0    0    0    0    0    0    0    0    0    0
   0    0    0    0    0    0    0    0    0    0  x24  x34  x44  x54    0    0    0    0    0    0    0    1    0    0    0    0    0    0    0    0    0    0    0    0    0    0    0    0    0    0    0
   0    0    0    0    0    0    0    0    0    0  x25  x35  x45  x55    0    0    0    0    0    0    0    0    1    0    0    0    0    0    0    0    0    0    0    0    0    0    0    0    0    0    0
   0    0    0    0    0    0 -x21 -x31 -x41 -x51 -x22 -x32 -x42 -x52    0    0    0    0    0    0    0    0    0    1    0    0    0    0    0    0    0    0    0    0    0    0    0    0    0    0    0
   0    0    0    0    0    0    0    0    0    0    0    0    0    0  x34    0    0    0    0    0    0    0    0    0    1    0    0    0    0    0    0    0    0    0    0    0    0    0    0    0    0
   0    0    0    0    0    0    0    0    0    0    0    0    0    0  x35    0    0    0    0    0    0    0    0    0    0    1    0    0    0    0    0    0    0    0    0    0    0    0    0    0    0
   0    0    0    0    0    0 -x21 -x31 -x41 -x51 -x22 -x32 -x42 -x52 -x33    0    0    0    0    0    0    0    0    0    0    0    1    0    0    0    0    0    0    0    0    0    0    0    0    0    0
   0    0    0    0    0    0    0    0    0    0    0    0    0    0    0  x22  x32  x42  x52    0    0    0    0    0    0    0    0    1    0    0    0    0    0    0    0    0    0    0    0    0    0
   0    0    0    0    0    0    0    0    0    0    0    0    0    0    0  x23  x33  x43  x53    0    0    0    0    0    0    0    0    0    1    0    0    0    0    0    0    0    0    0    0    0    0
   0    0    0    0    0    0    0    0    0    0    0    0    0    0    0  x24  x34  x44  x54    0    0    0    0    0    0    0    0    0    0    1    0    0    0    0    0    0    0    0    0    0    0
   0    0    0    0    0    0    0    0    0    0    0    0    0    0    0  x25  x35  x45  x55    0    0    0    0    0    0    0    0    0    0    0    1    0    0    0    0    0    0    0    0    0    0
   0    0    0    0    0    0    0    0    0    0    0    0    0    0    0 -x21 -x31 -x41 -x51    0    0    0    0    0    0    0    0    0    0    0    0    1    0    0    0    0    0    0    0    0    0
   0    0    0    0    0    0    0    0    0    0    0    0    0    0    0    0    0    0    0  x23  x33  x43  x53    0    0    0    0    0    0    0    0    0    1    0    0    0    0    0    0    0    0
   0    0    0    0    0    0    0    0    0    0    0    0    0    0    0    0    0    0    0  x24  x34  x44  x54    0    0    0    0    0    0    0    0    0    0    1    0    0    0    0    0    0    0
   0    0    0    0    0    0    0    0    0    0    0    0    0    0    0    0    0    0    0  x25  x35  x45  x55    0    0    0    0    0    0    0    0    0    0    0    1    0    0    0    0    0    0
   0    0    0    0    0    0    0    0    0    0    0    0    0    0    0 -x21 -x31 -x41 -x51 -x22 -x32 -x42 -x52    0    0    0    0    0    0    0    0    0    0    0    0    1    0    0    0    0    0
   0    0    0    0    0    0    0    0    0    0    0    0    0    0    0    0    0    0    0    0    0    0    0  x34  x44  x54    0    0    0    0    0    0    0    0    0    0    1    0    0    0    0
   0    0    0    0    0    0    0    0    0    0    0    0    0    0    0    0    0    0    0    0    0    0    0  x35  x45  x55    0    0    0    0    0    0    0    0    0    0    0    1    0    0    0
   0    0    0    0    0    0    0    0    0    0    0    0    0    0    0 -x21 -x31 -x41 -x51 -x22 -x32 -x42 -x52 -x33 -x43 -x53    0    0    0    0    0    0    0    0    0    0    0    0    1    0    0
   0    0    0    0    0    0    0    0    0    0    0    0    0    0    0    0    0    0    0    0    0    0    0    0    0    0  x45    0    0    0    0    0    0    0    0    0    0    0    0    1    0
   0    0    0    0    0    0    0    0    0    0    0    0    0    0    0 -x21 -x31 -x41 -x51 -x22 -x32 -x42 -x52 -x33 -x43 -x53 -x44    0    0    0    0    0    0    0    0    0    0    0    0    0    1
\end{verbatim}
\end{minipage}\end{lrbox}
\hspace{-0.5cm}\resizebox{0.75\textwidth}{!}{\usebox\myv}

}

\medskip

See the ancillary files for larger values of $m$.
 
\section{Iterated matrix multiplication and ABP's}\label{immabpsect}

The following result, while \lq\lq known to the experts\rq\rq , is not easily accessible in the literature. Moreover,
we give a precise formulation to facilitate measuring benchmark progress in different models.

In the following theorem note that $\himmc$ and $\dlabpc$ are only defined for homogeneous polynomials.
\begin{theorem}\label{cmeasures}
The complexity measures $\rdc$, $\dc$, $\labpc$, $\immc$, $\abpc$, $\himmc$, and $\dlabpc$ are all polynomially related.
More precisely, let $P$ be any polynomial. Let $\varphi(m):=\frac{m^3}3-\frac m 3 +2$ denote the layered ABP size of the Mahajan-Vinay construction for $\det_m$. Then
\begin{enumerate}
\item \label{itemd} $\dc(P)\leq \labpc(P)-1$. If $P$ has no constant part, then $\rdc(P)\leq \labpc(P)-1$.
\item\label{itema} $\labpc(P) \leq \varphi(\dc(P))$.
\item\label{itemg}
 By definition $\dc(P) \leq \rdc(P)$.
If $P$ has no constant part, then
$\rdc(P) \leq \varphi(\dc(P))-1$. If $\tcodim(P_{sing})\geq 5$, then $\rdc(P)=\dc(P)$.
\item \label{itemb}   $\labpc(P)= \immc(P)+1$. If $P$ is homogeneous, then $\dlabpc(P)= \himmc(P)+1$.
\item \label{iteme} 
 By definition  $\abpc(P) \leq \labpc(P) \leq \dlabpc(P)$, where $\dlabpc(P)$ is defined only if $P$ is homogeneous.
If $P$ is homogeneous of degree $d$ then $\dlabpc(P) \leq (d+1)\abpc(P)$.
\end{enumerate}
\end{theorem}

\begin{remark} It is an important and perhaps tractable open problem to prove
an $\o(m^2)$ lower bound for $\dc(\tperm_m)$. By Theorem \ref{cmeasures}, it would
suffice to prove an $\o(m^6)$ lower bound for $\himmc(\tperm_m)$. 
\end{remark}
 
\begin{remark} The computation model of homogeneous iterated matrix multiplication
 has the advantage that one is comparing the homogeneous iterated matrix multiplication polynomial himm 
 directly with the permanent, whereas with the determinant $\tdet_n$, one must compare with the
 {\it padded permanent} $\ell^{n-m}\tperm_m$.
 The padding causes insurmountable problems if one wants to
 find {\it occurrence obstructions} in the sense of \cite{MS1,MS2}. The problem  was first observed
 in \cite{MR3169697} and then proved insurmountable in \cite{2015arXiv151203798I} and \cite{DBLP:journals/corr/BurgisserIP16}. Thus {\it a priori}
 it might be possible to prove Valiant's conjecture via occurrence obstructions in the himmc model. However,
 with the determinant already one needed to understand difficult properties about three factor Kronecker coefficients,
 and for the himmc model, one would need to prove results about $m$-factor Kronecker coefficients, which are
 not at all understood.
 
 Regarding the geometric search for separating equations,  the advantage one gains by removing the padding is
 offset by the disadvantage of dealing with the himmc polynomial that for all known equations such
 as Young flattenings (which includes the method of shifted partial derivatives as a special case)
 and equations for degenerate dual varieties,  behaves far more generically than the determinant.
 \end{remark}

\begin{remark} One can also show that if $P$ is any polynomial of degree $d$, then $\labpc(P)\leq d(\abpc(P)^2)$.
\end{remark}

\begin{remark} Another complexity measure is the {\it homogeneous matrix powering complexity}:
If $P=\tr(A^m)$, then $P=\tr(A\cdot A\cdot \cdots \cdot A)$, thus $\himmc(P) \leq m \cdot \hmpc(P)$.

Conversely, if $\himmc(P)=n$, then $\dlabpc(P)=n+1$, so there exists a degree layered APB $\Gamma$ of size $n+1$ with value $P$.
Since all paths in $\Gamma$ from the source to the sink have exactly length $m$ we can identify the source and the sink
and get a directed graph $\Gamma'$ in which all closed directed walks have length exactly $m$.
These closed walks are in bijection to paths from the source to the sink in $\Gamma$.
Let $A$ be the $n \times n$ adjacency matrix of $\Gamma'$. We can interpret $\tr(A^m)$ as the sum over all closed directed walks of length exactly $m$ in $\Gamma'$,
where the value of each walk is the product of its edge weights.
We conclude that $P=\tr(A^m)$ and thus $\hmpc(P) \leq \himmc(P)$.
\end{remark}

\begin{proof}[ {Proof of Theorem \ref{cmeasures}}]
\eqref{itemd} is Proposition~\ref{pro:constructadjmatrix}.

Proof of \eqref{itema}: We first write the determinant polynomial $\det_{\dc(P)}$ as a size $\varphi(\dc(P))$ layered ABP $\Gamma$ using \ref{pro:MahVin}.
The projection that maps $\det_{\dc(P)}$ to $P$ can now be applied to $\Gamma$ to yield a size $\varphi(\dc(P))$ layered ABP of $P$.

Proof of \eqref{itemg}: To see the second inequality we combine \eqref{itemd} and \eqref{itema}.
 The last assertion is von zur Gathen's result \cite{MR910987}. 
 
Proof of \eqref{itemb}:
We prove $\labpc(P) \leq \immc(P)+1$.
Given $n_1,\ldots,n_m$ with $n_1=1$ and $n_1+\cdots+n_m=\immc(P)$
and linear maps $B_j$, $1 \leq j \leq m$,
we construct the ABP $\G$ that has a single vertex at level $m+1$, $n_j$ vertices at level $j$, $1\leq j\leq m$,
and is the complete bipartite graph between levels. The labels of $\Gamma$ are given by the $B_j$.
We now prove $\immc(P) \leq \labpc(P)-1$.
Given a layered ABP $\G$ with $m+1$ layers,
recall that by definition  $\G$ has only 1 vertex in the top layer and only one vertex in the bottom layer.
 Let  $n_j$ denote  the number of vertices in layer $j$, $1 \leq j \leq m$.
 Define the linear maps $B_j$ by reading off the labels between layer $j$ and layer $j+1$.
 The  proof of the second claim   is    analogous.

Proof of \eqref{iteme}:  {(This argument was outlined in \cite{Nisan:1991:LBN:103418.103462}.)} We first homogenize and then adjust the ABP.
Replace each vertex $v$ other than $s$ by $d+1$ vertices $v^1,v^2,\ldots,v^{d+1}$ corresponding to the homogeneous parts of $\Gamma_v$.
Replace each edge $e$ going from a vertex $v$ to a vertex $w$ by $(2d+1)$ edges, where we split the linear and constant parts:
If $e$ is labeled by $\ell+\delta$, where $\ell$ is linear and $\delta\in\BC$,
the edge from $v^i$ to $w^i$, $1 \leq i \leq d$, is labeled with $\delta$ and the edge from $v^i$ to $w^{i+1}$, $1 \leq i \leq d-1$,
is labeled with $\ell$.
We now have a homogeneous ABP.
Our task is to make it degree layered.
As a first approach we assign each degree $i$ vertex to be in layer $i$,
but there may be edges labeled with constants between vertices in the same layer.
The edges between vertices of different layers are linear forms.
Call the vertices in layer $i$ that have edges incoming from layer $i-1$, {\it layer $i$ entry vertices}.
Remove the non-entry vertices.
>From entry vertex of layer $i$ to entry vertex of layer $i+1$, use the {\it linear form} computed by the sub-ABP between them.
In other words, for every pair $(v,w)$ of layer $i$ entry vertex $v$ and layer $i+1$ entry vertex $w$, put an edge from $v$
to $w$ with weight
$$
\sum_p \Pi_e \twt(e)
$$
where the sum is over paths $p$ from $v$ to $w$ and the product is over edges in the path $p$. 
The resulting ABP is degree homogeneous and computes $P$. 
\end{proof}

\section{Proofs of Theorems \ref{expthm} and \ref{nosqueezethm}}\label{expthmpfsect}

 {The following arguments appeared in \cite{Nisan:1991:LBN:103418.103462} in slightly different language.
We reproduce them in  the language of this paper for convenience.}

\begin{proof}[Proof of Theorem \ref{expthm}]
This can be seen directly from a consideration about evaluation dimension that we explain now.
We prove the stronger statement that the degree homogeneous ABP must have at least $\binom{m}{s}$ vertices at layer $s$, $0 \leq s \leq m$.
Summing up the binomial coefficients and using Theorem~\ref{cmeasures}\eqref{itemb} yields the result.

We consider the degree homogeneous ABP $\G$ with $m+1$ layers that computes $\det_m$ (or $\per_m$).
Keeping the labels from the source to layer $s$ and setting the labels on all other layers to constants we see that
all terms of the form
$\sum_{\s\in \FS_{m}} c_{\s} y^{1, \s(1)}\cdots y^{s,\s(s)}$
can be computed
by taking linear combinations of the polynomials $\G_v$, where $v$ is a vertex in layer $s$.
Since these terms span a vector space of dimension $\binom{m}{s}$
there must be at least $\binom{m}{s}$ linearly independent polynomials $\G_v$, so
there must be at least $\binom{m}{s}$ vertices on layer $s$.

The Grenet determinantal presentation of $\tperm_m$ \cite{Gre11} and the 
regular determinantal presentation of $\tdet_m$ of \cite{LRpermdet}
give rise to column-wise multilinear IMM presentations of size $2^m-1$.
\end{proof}

\begin{proof}[Proof of Theorem \ref{nosqueezethm}] 
 {The proof is essentially the same as the proof of Theorem
\ref{expthm}. Without loss of generality assume
it is the upper left  $L\times R$ sub-matrix appearing in  
the first $L$ terms. The terms of the form
$\sum_{\s\in \FS_{R} }c_{\s} y^{1}_{\s(1)}\cdots y^{L}_{\s(L)}$,
with the $c_{\s}$ nonzero constants, 
all appear in $\tdet_m$ and $\tperm_m$,  so they must appear independently
in the row vector $A_{L}\cdots A_1$. There are 
$\binom{R}{L}$ such terms so we conclude.}
\end{proof}

\section{Proof of Theorem \ref{rankone}}\label{rankksect}

\subsection{Regular determinantal expressions}\label{rdcsect}

 For $P\in S^m\BC^M$ define the {\it symmetry group of $P$}:
$$
G_P:=\{ g\in GL_M\mid P(g\cdot y)=P(y) \ \forall y\in \BC^M\}
$$

The group 
$G_{\tdet_n}$ essentially consists of multiplying an $n\times n$ matrix  $X$ on the left and right by matrices
of determinant one, and the transpose map, $X\mapsto X^T$.
Using $G_{\tdet_n}$,  without loss of
generality we may assume $\Lambda$ in a regular determinantal expression is the identity matrix
except with the $(1,1)$-entry set equal to zero. We call a regular determinantal representation
{\it standard} if $\Lambda$ is so normalized.

Let the upper indices stand for variable names (i.e.\ positions in a small $m \times m$ matrix)
and the lower indices stand for positions in a big $n \times n$ matrix.
If $A$ is an $n \times n$ matrix whose entries are affine linear forms in $m^2$ variables, then we write
\[
A = \Lambda + y^{1,1} X^{1,1} + y^{1,2} X^{1,2} + \cdots + y^{m,m} X^{m,m}
\]
with $m^2+1$ matrices $\Lambda, X^{1,1}, X^{1,2},\ldots,X^{m,m}$ of format  {$n \times n$}.

\begin{lemma}\label{lem:properties}
If $\det(A)\in\{\pm \det_m,\pm \per_m\}$ and $\Lambda$ is standard, then
\begin{enumerate}
\item[(I)] $A_{1,1}=0$,
\item[(II)] $\sum_{j=2}^n A_{1,j} A_{j,1} = 0$
\item[(III)] In the first column of $A$ there are  at least $m$ different entries. The same holds for the first row of $A$.
\end{enumerate}

\end{lemma}
\begin{proof}
As observed in  \cite{2015arXiv150502205A}, (I) and (II) hold in any regular determinantal expression for a homogeneous
polynomial of degree $m\geq 2$ with standard $\Lambda$.
To prove (III),  
by \cite{MR0029360} (resp. \cite{MR3340547}) $\{\tdet_m=0\}\subset \BC^{m^2}$ (resp. $\{\tperm_m=0\}\subset \BC^{m^2}$)
does not admit a linear subspace of dimension $m(m-1)+1$. This implies that neither polynomial admits
an expression of the form $\ell_1p_1+\cdots +\ell_{m-1}p_{m-1}$ with $\ell_j$ linear and $p_j$ of degree $m-1$, as otherwise
the common zero set of $\ell_1\hd \ell_{m-1}$ would provide a linear space of dimension $m(m-1)+1$ on the hypersurface.
If we have a regular determinantal expression of $\tperm_m$ or $\tdet_m$, this implies that at least $m$ different
linear forms     appear  in the first column of $X$ and at least $m$ different linear forms    appear  in the first row of $X$.
\end{proof}

We are free to change our determinantal expression by elements of  the group $G_{\tdet_n,\Lambda}$ preserving
both $\tdet_n$ and $\Lambda$, which by \cite{LRpermdet} is, for $M\in \Mat_{n\times n}(\BC)$:
$$
\{M\mapsto
\begin{pmatrix} \l  & 0\\ v & g\end{pmatrix}
M \begin{pmatrix} 1 & w^T\\ 0& g\end{pmatrix}^{-1}
\mid g\in\GL_{n-1}, v \in \BC^{n-1}, w\in \BC^{n-1},\l \in \BC^*\}\cdot
\langle\transp\rangle,
$$
Where $\langle \transp\rangle\simeq \BZ_2$ is the group generated by transpose.

\subsection{Rank one  regular determinantal expressions}

Theorem \ref{rankone} will follow from Lemmas \ref{redlem} and \ref{threelem}.

\begin{lemma}\label{redlem}
Let $P_m\in S^m(Mat_{m\times m})$ be the permanent or determinant.
\begin{enumerate}
\item If $P_{m_0}$ does not admit a rank $k$ determinantal expression, then $P_{m}$ does not
admit a rank $k$ determinantal expression for all $m\geq m_0$.
\item If $P_{m_0}$ does not admit a rank $k$ \emph{regular} determinantal expression, then $P_{m}$ does not
admit a rank $k$ \emph{regular} determinantal expression for all $m\geq m_0$.
\end{enumerate}
\end{lemma}
\begin{proof}
 Without loss of generality  $m=m_0+1$.
Say $P_{m}$ admitted a rank $k$ $n \times n$ determinantal expression $A=\Lambda+\sum_{i,j=1}^{m} X^{i,j}y^{i,j}$.
 Set $y^{m,u}=y^{v,m}=0$ for $1\leq u,v,\leq m_0=m-1$.
We obtain  the matrix
$\Lambda+  X^{m,m}y^{m,m} + \sum_{u,v=1}^{m_0} X^{u,v}y^{u,v}$.
This yields a rank $k$ determinantal expression for $y^{m,m} \cdot P_{m_0}$, which proves the first part if we set $y^{m,m} = 1$. 

For the second part, first note that every determinantal expression 
$P_{m_0} = \det(\Lambda'+\sum_{u,v=1}^{m_0} X^{u,v}y^{u,v})$
satisfies   $\trank \Lambda' \leq n-1$  because $P_{m_0}$ has no constant part.
Thus to prove that a determinantal expression for $P_{m_0}$ is regular it suffices to show that $\trank \Lambda' \geq n-1$.

Say $P_{m}$ admitted a rank $k$ $n \times n$ regular determinantal expression $A=\Lambda+\sum_{i,j=1}^{m} X^{i,j}y^{i,j}$,
so $\trank \Lambda = n-1$.
Then  
$\trank(\Lambda+y^{m,m}_0 X^{m,m}) \geq n-1$ for   almost all $y^{m,m}_0 \in \BC$.
Choosing such a  $y^{m,m}_0 \neq 0$  we obtain a regular determinantal expression for $y^{m,m}_0 \cdot P_{m_0}$.
Rescaling the first rows of $\Lambda$ and all $X^{i,j}$ with $\frac 1 {y^{m,m}_0}$ we get a regular determinantal expression for $P_{m_0}$.
\end{proof}

\medskip

\begin{lemma} \label{threelem} Neither $\tdet_3$ nor $\tperm_3$ admits a rank one regular determinantal representation.
\end{lemma}

The idea of the proof is simple: each monomial in the expression of $\tperm_3$ (or $\tdet_3$)
 must have a contribution from the first column and the first row, say slots
 $(s,1)$ and $(1,t)$. But then to have a homogeneous degree three expression, the third variable in the monomial
 must appear in the $(t,s)$-slot. This is sufficiently restrictive that one can conclude. Now
 for the details:

 \subsection{Proof of Lemma \ref{threelem}}
 
 Before proving the Lemma, we establish some preliminary results.
 
\newcommand{\fref}[1]{\ref{#1}}

\newcommand{\partinto}[1][]{\smash{\mathord{\mathchoice{%
  \xymatrix@=0.4em@1{%
  \ar@{|-}[rr]_-*--{\scriptstyle #1}
  &*{\phantom{\scriptstyle{#1}}}&}
}{
  \xymatrix@=0.25em@1{%
  \ar@{|-}[rr]_-*--{\scriptstyle #1}
  &*{\phantom{\scriptstyle{#1}}}&}
}{
  \xymatrix@=0.2em@1{%
  \ar@{|-}[rr]_-*--{\scriptscriptstyle #1}
  &*{\phantom{\scriptscriptstyle{#1}}}&}
}{}}}}
\newcommand{\partintonosmash}[1][]{\mathord{\mathchoice{%
  \xymatrix@=0.4em@1{%
  \ar@{|-}[rr]_-*--{\scriptstyle #1}
  &*{\phantom{\scriptstyle{#1}}}&}
}{
  \xymatrix@=0.25em@1{%
  \ar@{|-}[rr]_-*--{\scriptstyle #1}
  &*{\phantom{\scriptstyle{#1}}}&}
}{
  \xymatrix@=0.2em@1{%
  \ar@{|-}[rr]_-*--{\scriptscriptstyle #1}
  &*{\phantom{\scriptscriptstyle{#1}}}&}
}{}}}
\newcommand{\partintostar}[1][]{\smash{\mathord{\mathchoice{%
  \xymatrix@=0.4em@1{%
  \ar@{|-}[rr]_-*--{\scriptstyle #1}^-*--{\scriptstyle \ast}
  &*{\phantom{\scriptstyle{#1}}}&}
}{
  \xymatrix@=0.25em@1{%
  \ar@{|-}[rr]_-*--{\scriptstyle #1}^-*--{\scriptstyle \ast}
  &*{\phantom{\scriptstyle{#1}}}&}
}{
  \xymatrix@=0.2em@1{%
  \ar@{|-}[rr]_-*--{\scriptscriptstyle #1}^-*--{\scriptstyle \ast}
  &*{\phantom{\scriptscriptstyle{#1}}}&}
}{}}}}
\newcommand{\partintostarnosmash}[1][]{\mathord{\mathchoice{%
  \xymatrix@=0.4em@1{%
  \ar@{|-}[rr]_-*--{\scriptstyle #1}^-*--{\scriptstyle \ast}
  &*{\phantom{\scriptstyle{#1}}}&}
}{
  \xymatrix@=0.25em@1{%
  \ar@{|-}[rr]_-*--{\scriptstyle #1}^-*--{\scriptstyle \ast}
  &*{\phantom{\scriptstyle{#1}}}&}
}{
  \xymatrix@=0.2em@1{%
  \ar@{|-}[rr]_-*--{\scriptscriptstyle #1}^-*--{\scriptstyle \ast}
  &*{\phantom{\scriptscriptstyle{#1}}}&}
}{}}}

\newcommand{\kron}[3]{k({#1};{#2};{#3})}

 \newcommand{\IC}{\ensuremath{\mathbb{C}}}
\newcommand{\IA}{\ensuremath{\mathbb{A}}}
\newcommand{\IN}{\ensuremath{\mathbb{N}}}
\newcommand{\IZ}{\ensuremath{\mathbb{Z}}}
\newcommand{\IQ}{\ensuremath{\mathbb{Q}}}
\newcommand{\IR}{\ensuremath{\mathbb{R}}}
\newcommand{\IIF}{\ensuremath{\mathbb{F}}}
\newcommand{\IP}{\ensuremath{\mathbb{P}}}
\newcommand{\aS}{\ensuremath{\mathfrak{S}}}
\newcommand{\aG}{\ensuremath{\mathfrak{G}}}
\newcommand{\aE}{\ensuremath{\mathsf{E}}}
\newcommand{\sS}{\ensuremath{\mathscr{S}}}
\newcommand{\sF}{\ensuremath{\mathscr{F}}}
\newcommand{\sV}{\ensuremath{\mathscr{V}}}
\newcommand{\sW}{\ensuremath{\mathscr{W}}}

\newcommand\Zehn{10}
\newcommand\Elf{11}
\newcommand\Zwoelf{12}

\newcommand\rk{\textup{rk}}

\raggedbottom
 
\sloppy

\begin{lemma}\label{lem:thirdposition}
Let $\det(A)\in\{\pm \det_3,\pm \per_3\}$ and let $\Lambda$ be standard.
Let $1 \leq i_1,j_1,i_2,j_2,i_3,j_3 \leq 3$.
If the monomial $y^{i_1,j_1}\cdot y^{i_2,j_2}\cdot y^{i_3,j_3}$ appears in $\det(A)$, then there exists a permutation $\pi \in \FS_3$
and integers $2 \leq k,\ell\leq n$, $k \neq \ell$ such that
$X^{i_{\pi(1)},j_{\pi(1)}}_{k,1} \neq 0$,
$X^{i_{\pi(2)},j_{\pi(2)}}_{1,\ell} \neq 0$, and
$X^{i_{\pi(3)},j_{\pi(3)}}_{\ell,k} \neq 0$.
\end{lemma}
\begin{proof}
By Lemma~\ref{lem:properties}(I) we have $A_{1,1}=0$.
For subsets $L,K\subseteq \{1,\ldots,n\}$
let $A(L,K)$ denote the matrix that results from $A$ by striking out the rows $L$ and the columns $K$.
In $A$ set all variables to zero besides $y^{i_1,j_1}$,$ y^{i_2,j_2}$, and $y^{i_3,j_3}$ and call the resulting matrix $B$.
Since $\det(A)$ is homogeneous of degree 3, every other monomial in $\det(A)$ involves one of the variables that were set to zero.
 Hence    $\det(B)=y^{i_1,j_1}\cdot y^{i_2,j_2}\cdot y^{i_3,j_3}$.
In particular $\det(B)\neq 0$.
Since $\Lambda$ has only zeros in the first row,
we conclude that there exists a nonzero variable entry in the first row of $B$ (in column $2,\ldots,n$),
w.l.o.g.\ $X^{i_{2},j_{2}}_{1,\ell} \neq 0$,
whose minor $\det(B(\{1\},\{\ell\}))$ contains the summand $y^{i_1,j_1}y^{i_3,j_3}$.
Since $B(\{1\},\{\ell\})$ has no constant terms in the first column,
a variable $y^{i_1,j_1}$ or $y^{i_3,j_3}$ must appear in the first column of $B(\{1\},\{\ell\})$, w.l.o.g.\ $X^{i_{1},j_{1}}_{k,1} \neq 0$,
such that its minor $\det(B(\{1,k\},\{\ell,1\}))$ contains the summand $y^{i_3,j_3}$.

 Assume  for a moment that $k = \ell$, i.e., in the first column no other position has a $y^{i_1,j_1}$
and in the first row no other position has a $y^{i_2,j_2}$.
This is impossible due to Lemma~\ref{lem:properties}(II).

Finally assume $k \neq \ell$.
Since the constant part of $B(\{1,k\},\{\ell,1\})$ is a permutation matrix with a single hole,
this hole is where $B(\{1,k\},\{\ell,1\})$ must have a nonzero entry $y^{i_3,j_3}$.
In $A$ this is at position $(\ell,k)$. 
\end{proof}

We now give names to some standard operations on matrices that we will use in the upcoming arguments.
 We continue to assume $\Lambda$ is standard. 
\begin{itemize}
\item Adding/subtracting a multiple of the first column of $A$ to other columns of $A$ is called a \emph{first column operation}.
Analogously for \emph{first row operations}.
First row or first column operations belong to $G_{\tdet_n,\Lambda}$. 
\item If we add/subtract multiples of other rows/columns from each other we call this a \emph{Gauss-Jordan operation}.
 Gauss-Jordan operations belong to $G_{\tdet_n}$ but not $G_{\tdet_n,\Lambda}$. 
 \item Let $2 \leq i,j \leq n$.
Permuting rows $i$ and $j$ and then permuting columns $i$ and $j$ is called a \emph{permutation conjugation}.
 Permutation conjugations belong to $G_{\tdet_n,\Lambda}$. 
 \item Let $2 \leq i,j \leq n$.
For $\alpha \in \IC$,
adding $\alpha$ times the $i$th row to the $j$th row of $A$
and then subtracting $\alpha$ times the $j$th column from the $i$th column of $A$
is called a \emph{elimination conjugation}.
 Elimination conjugations belong to $G_{\tdet_n,\Lambda}$.
 \end{itemize}
 
We are now ready to prove Lemma \ref{threelem}.
We assume the contrary and let $P = \det_3$ or $P=\per_3$ such that
\begin{enumerate}[(1)]
\item $A$ is an $n\times n$ matrix,
\item $\det(A) \in \{-P,P\}$,
\item $\rk\Lambda=n-1$,
\item $\rk(X^{i,j})=1$ for all $1 \leq i,j \leq 3$.
\end{enumerate}

Note that the operations defined above all preserve (1)-(4).
Thus performing a Gauss-Jordan elimination on $\Lambda$ (and performing the operations on the whole matrix $A$) we can make $\Lambda$ standard while preserving (1)-(4).
So we can additionally assume:

\begin{enumerate}
\item[(5)]  $\Lambda$ is standard and hence    properties (I),(II),(III)  from Lemma~\ref{lem:properties} hold.
\end{enumerate}

Using (5)(III) we pick a variable that appears in $A$ in the first column. It cannot appear at position (1,1) because of (5)(I).

The operation of \emph{permuting variable names} by permuting rows and/or columns of the $3 \times 3$ variable matrix
preserves (1)-(5)  and belongs to $G_{\tperm_3}$. Doing so we can assume that $X^{1,1}$ has a nonzero entry in column 1, not in position (1,1).
Using permutation conjugation we can move this position to position (2,1).
Using first column operations we can make $X^{1,1}$ have only zeros in row 2, besides the nonzero entry at position (2,1).
Using elimination conjugation we can make $X^{1,1}$ have only zeros in column 1, besides the nonzero entry at position (2,1).
Using (4) we see that $X^{1,1}$ only has a single nonzero entry: at position (2,1).
So besides (1)-(5) we can assume:
\begin{enumerate}

\item [(6)] $X^{1,1}_{i,j}\neq 0$ iff $(i,j)=(2,1)$.
\end{enumerate}
Combining (5)(II) with (6) it follows that
\begin{equation}\tag{6b}
A_{1,2}=0.
\end{equation}

We want to deduce more facts about $A$ by setting several variables to zero.
Set all variables in $A$ to zero besides $y^{1,1}$, $y^{2,2}$, $y^{3,3}$ and 
call the resulting matrix $B$. From (2) it follows that we have
\begin{equation}\tag{2b}
\det(B)=\pm y^{1,1}y^{2,2}y^{3,3}.
\end{equation}
By (2b)
the first row of $B$ cannot be all zeros, so by (6) and the standardness granted by (5) we have that $X^{2,2}$ or $X^{3,3}$ have a nonzero entry in the first row.
If $X^{2,2}$ has a nonzero entry, we permute the 2nd and 3rd row and column in the $3 \times 3$ variable matrix.
This operation preserves (1)-(6), so we conclude that we can assume
\begin{enumerate}
\item [(7)] $X^{3,3}$ has a nonzero entry in the first row.
\end{enumerate}
Combining (4) and (5)(I) it follows that
\begin{equation}\tag{7b}
\text{The first column of  } X^{3,3} \text{ is zero.} 
\end{equation}
Using permutation conjugation we want to move the nonzero entry from (7) in $X^{3,3}$ to position $(1,n)$.
Note that according to (5)(I) and (6b) this entry is in row 1 in some column $3,\ldots,n$.
Permutation conjugation on indices $3,\ldots,n$ preserves (1)-(7).
Thus we can use permutation conjugations to assume that
\begin{enumerate}
\item [(8)] $X^{3,3}_{1,n}\neq 0$.
\end{enumerate}
Using first row operations preserves (1)-(8), for example they preserve (6) because of (5)(I).
Thus we can use first row operations to assume that
\begin{enumerate}
\item [(9)]  The only nonzero entry of $X^{3,3}$ in column $n$ is $(1,n)$. 
\end{enumerate}
Elimination conjugation (adding $\alpha$ times column $n$ to column $3 \leq k \leq n-1$ and then subtracting $\alpha$ times row $k$ from row $n$)
preserves (1)-(9). We use these operations together with (5)(I) and (6b) to assume that
\begin{enumerate}
\item [(10)]  The only nonzero entry of $X^{3,3}$ in row  $1$ is $(1,n)$.  
\end{enumerate}
Combining (4) with (9) and (10) we conclude
\begin{equation}\tag{10b}
X^{3,3}_{i,j}\neq 0 \text{ iff } (i,j)= (1,n).
\end{equation}
With (5)(II) we conclude
\begin{equation}\tag{10c}
A_{n,1}=0.
\end{equation}
Let $A'$ denote the submatrix of $A$ obtained by deleting the rows $1$ and $2$ and the columns $1$ and $n$.
By assumption $\det(A)$ has a summand $y^{1,1}y^{2,2}y^{3,3}$.
Using (6) and (10b), a double Laplace expansion implies that $\det(A')$ has a term $y^{2,2}$.
By the standardness granted by (5), the homogeneous degree 1 part of $\det(A')$ is precisely the entry at position $(n,2)$ in $A$.
It follows that
\begin{equation}\tag{10d}
X^{2,2}_{n,2} \neq 0.
\end{equation}
We claim that
\begin{equation}\tag{10e}
\text{$X^{i,j}_{n,2}=0$ for all $(i,j)\neq(2,2)$}.
\end{equation}
Assume that $X^{i,j}_{n,2}\neq 0$ for some $1 \leq i,j \leq 3$.
Set all variables in $A$ to zero but $y^{1,1}$, $y^{3,3}$, and $y^{i,j}$, and call the resulting matrix $E$.
Since $y^{1,1}$ and $y^{3,3}$ appear only once in $A$ and since $\Lambda$ is standard, the degree 3 part of $\det(E)$
contains all summands that appear in $y^{1,1}\cdot y^{3,3} \cdot q$,
where $q$ is the linear part of $\det(A(\{1,2\},\{1,n\}))$.
Indeed, $q$ equals the linear part of $A$ at position $(n,2)$.
Since $\det(A)\in\{\det_3,\per_3\}$ it follows that $q=y^{2,2}$, thus
$(i,j)=(2,2)$. This proves the claim (10e).

 We   deduce more facts about $A$ by setting several other variables to zero.
Set all variables in $A$ to zero besides $y^{1,1}$, $y^{2,3}$, $y^{3,2}$ and 
call the resulting matrix $C$. From (2) it follows that we have
\begin{equation}\tag{2c}
\det(C)=\pm y^{1,1}y^{2,3}y^{3,2}.
\end{equation}
By (2c)
the first row of $C$ cannot be all zeros, so by (5) and (6) we have that $X^{2,3}$ or $X^{3,2}$ have a nonzero entry in the first row.
If it is $X^{3,2}$ and not $X^{2,3}$, then we can apply the transposition
 from $G_{\tperm_3}$  (preserving (1)-(10) because $X^{1,1}$, $X^{2,2}$, and $X^{3,3}$ are fixed) to ensure:
\begin{enumerate}
 \item [(11)] $X^{2,3}$ has at least one nonzero entry in row 1.
\end{enumerate}
Combining (11) and (4) and (5)(I) we see that
\begin{equation}\tag{11b}
X^{2,3} \text{ is zero in the first column}.
\end{equation}

There are two cases:
\subsection*{Case 1: In row 1, $X^{2,3}$ is nonzero only in column $n$}
We will show that this case cannot appear.

>From the assumption of case 1 we conclude with (4) that
\begin{equation}\tag{11$'$}
X^{2,3} \text{ is zero everywhere but in the last column}.
\end{equation}
 We  apply Lemma~\ref{lem:thirdposition} with the monomial $y^{1,1} y^{2,3} y^{3,2}$ that appears in $\det(A)$,
so
\begin{itemize}
\item one of the three variables goes to the first column in some row $k \neq 1$,
\item one goes to the first row in some column $\ell \neq 1$, 
\item and one goes to position $(\ell,k)$.
\end{itemize}
Since by (6) $y^{1,1}$ only appears in the first column, it must be the variable that goes to the first column.
Again, by (6) we have $k=2$.
By (11$'$) $y^{2,3}$ cannot go to the second column, in particular not to position $(\ell,k)$,
so $y^{2,3}$ goes to the first row.
By (11) and (11$'$), $y^{2,3}$ goes to position $(1,n)$.
Therefore $y^{3,2}$ goes to position $(n,2)$.
This is a contradiction to (10e).
We conclude that case 1 cannot appear.

\subsection*{Case 2: In row 1, $X^{2,3}$ is nonzero in some column which is not $n$}
Permutation conjugation on the indices $3,\ldots,n-1$ preserves (1)-(11).
By (5)(I) and (6b) and the case assumption
these permutation conjugations are sufficient to assume
\begin{enumerate}
\item [(12)] $X^{2,3}_{1,{n-1}}\neq 0$.
\end{enumerate}
Since elimination conjugations (subtracting multiples of column $n-1$ from columns $3,\ldots,n-2$
and then adding multiples of rows $3,\ldots,n-2$ to row $n-1$)
preserve (1)-(12)
we can assume that 
\begin{enumerate}
\item [(13)] In row 1, the only positions of nonzero entries in $X^{2,3}$ are $(1,n-1)$ and possibly additionally $(1,n)$.
\end{enumerate}
Using (4) we conclude
\begin{enumerate}
\item [(13b)] $X^{2,3}$ vanishes in columns $1,\ldots,n-2$.
\end{enumerate}
Using elimination conjugation (subtract a multiple of column $n-1$ from column $n$ and add a multiple of row $n$ to row $n-1$), 
which preserves (1)-(13),
we can assume that
\begin{enumerate}
\item [(14)] $X^{2,3}_{1,n}=0$.
\end{enumerate}
 Then (4),   (12), (13b),  and (14) imply   
 \begin{enumerate}
\item [(14b)] $X^{2,3}$ is nonzero only in column $n-1$.
\end{enumerate}
 Lemma~\ref{lem:thirdposition} applied to  the monomial $y^{1,1}y^{2,3}y^{3,2}$ gives \begin{itemize}
\item one of the three variables goes to the first column in some row $k \neq 1$,
\item one goes to the first row in some column $\ell \neq 1$, 
\item and one goes to position $(\ell,k)$.
\end{itemize}
Since by (6) $y^{1,1}$ only appears in the first column, $y^{1,1}$ must be the variable that goes to the first column.
 Again,   $k=2$  by (6). 
By (13b) $y^{2,3}$ cannot go to the second column, so $y^{2,3}$ goes to the first row.
By (14b) $y^{2,3}$ appears at position $(1,n-1)$.
Therefore $y^{3,2}$ goes to position $(n-1,2)$.
Summarizing:
\begin{enumerate}
\item [(14c)] $X^{2,3}_{1,n-1}\neq 0$ and $X^{3,2}_{n-1,2}\neq 0$.
\end{enumerate}
Using (14b) and (5)(II) we conclude
\begin{enumerate}
\item [(14d)]  $A_{n-1,1} = 0$.
\end{enumerate}

 Lemma~\ref{lem:thirdposition} applied to  the monomial $y^{1,2}y^{2,1}y^{3,3}$ gives 
\begin{itemize}
\item one of the three variables goes to the first column in some row $k \neq 1$,
\item one goes to the first row in some column $\ell \neq 1$, 
\item and one goes to position $(\ell,k)$.
\end{itemize}
Since the only position for $y^{3,3}$ is fixed, $y^{3,3}$ goes to the first row to position $(1,n)$, so $\ell=n$.
Using (10c) and (14d) we see $k \leq n-2$. Moreover (5)(I) says $k \neq 1$ and (10e) says $k\neq 2$. So in total we have $3 \leq k \leq n-2$.
Using permutation conjugation we can assume $k=3$, so that the only cases left to consider are:

\subsection*{Case 2.1: $X^{1,2}_{3,1}\neq 0$ and $X^{2,1}_{n,3}\neq 0$}
Using elimination conjugation we can get rid of any occurrences of $y^{1,2}$ at positions $(k,1)$ for $4 \leq k \leq n-2$.
So with (6b) and (5)(II) it follows
\begin{enumerate}
\item [(15)] $A_{1,3} = 0.$
\end{enumerate}
  Lemma~\ref{lem:thirdposition} applied to  the monomial $y^{1,2}y^{2,3}y^{3,1}$ gives 
\begin{itemize}
\item one of the three variables goes to the first column in some row $k \neq 1$,
\item one goes to the first row in some column $\ell \neq 1$, 
\item and one goes to position $(\ell,k)$.
\end{itemize}
By (14b), $y^{2,3}$ only appears in column $n-1$, so $y^{2,3}$ does not go to the first column.
We make a small case distinction:
First assume that $y^{2,3}$ does not go in the first row.
Then $y^{2,3}$ goes to position $(\ell,k)$ with $k=n-1$.
But $k=n-1$ is impossible because $A_{n-1,1}=0$ by (14d).

On the other hand, if we assume that $y^{2,3}$ goes in the first row, then $\ell=n-1$.
By (4) and the case assumption~2.1, since $\ell=n-1$, $y^{1,2}$ cannot go to $(\ell,k)$, so it must go in the first column.
Therefore $y^{3,1}$ goes to position $(n-1,k)$.
Since $A_{n-1,1}=0$ by (14d) and $A_{n,1}=0$ by (10c) and $X^{2,2}_{n,2}\neq 0$ by (10d)
the variables $y^{3,1}$ and $y^{2,2}$ cannot appear in column 1 because of (4).
Thus using Lemma~\ref{lem:thirdposition} for the monomial $y^{2,2}y^{1,3}y^{3,1}$
we see that $y^{1,3}$ must appear in the first column.
But for the sake of contradiction we now use Lemma~\ref{lem:thirdposition} for the monomial $y^{1,3}y^{3,2}y^{2,1}$ as follows:
We have $A_{1,1}=A_{1,2}=A_{1,3}=0$ by (5)(I) and (6b) and (15).
The variable $y^{1,3}$ appears in column~1,
the variable $y^{3,2}$ appears in column~2 by (14c),
and the variable $y^{2,1}$ appears in column~3 (case assumption 2.1).
Thus by (4) none of these three variables appears in row 1, which is a contradiction to Lemma~\ref{lem:thirdposition}.
Therefore case 2.1 cannot appear.

\subsection*{Case 2.2: $X^{2,1}_{3,1}\neq 0$ and $X^{1,2}_{n,3}\neq 0$}
Using elimination conjugation  we   get rid of any occurrences of $y^{2,1}$ at positions $(k,1)$ for $k \neq 3$:
\begin{enumerate}
\item [(15)] In the first column $y^{2,1}$ appears only at position $(3,1)$.
\end{enumerate}
So with (6b) and (5)(II) it follows
\begin{enumerate}
\item [(15b)] $A_{1,3} = 0.$
\end{enumerate}
  Lemma~\ref{lem:thirdposition}  applied to  the monomial $y^{1,2}y^{2,3}y^{3,1}$ gives 
\begin{itemize}
\item one of the three variables goes to the first column in some row $k \neq 1$,
\item one goes to the first row in some column $\ell \neq 1$, 
\item and one goes to position $(\ell,k)$.
\end{itemize}
Since $X^{1,2}_{n,3}\neq 0$ and
since (4) combined with (10c) and (15b) implies that
$A_{n,1}=A_{1,3}=0$, it follows that $y^{1,2}$ is the variable that appears at position $(\ell,k)$. 
Since by (14b) $y^{2,3}$ only appears in column $n-1$,
$y^{2,3}$ must be the variable that appears in the first row at position $(1,n-1)$. Thus $\ell=n-1$.
Moreover, the third variable $y^{3,1}$ must appear in the first column.

In the first column $y^{3,1}$ cannot appear in rows $1$, $n-1$, or $n$ by (5)(I), (10c), (14d).
We want to use elimination conjugation on rows/columns $2,\ldots,n-2$ to ensure that $y^{3,1}$ appears only once in the first column.
But not every operation preserves (1)-(15).

\subsection*{Case 2.2.1: In column 1 $y^{3,1}$ appears in a row $4 \leq j \leq n-2$}
If $y^{3,1}$ appears in column 1 in a row $4 \leq j \leq n-2$, then elimination conjugation can be used to
ensure that
\begin{enumerate}
\item [(16)] In column 1 $y^{3,1}$ appears only in row $j$.
\end{enumerate}
Thus $k=j$. Thus $y^{1,2}$ occurs at position $(\ell,k)=(n-1,j)$.
With (4) and with case assumption 2.2 we see that
\begin{enumerate}
\item [(16b)] $X^{1,2}_{n,j}\neq 0$.
\end{enumerate}
Since $y^{3,3}$ occurs only at position $(1,n)$ and $\Lambda$ is zero in the first row,
$y^{3,1}y^{1,2}y^{3,3}$ occurs in $\det(A)$ iff $y^{3,1}y^{1,2}$ occurs in $\det(A(\{1\},\{n\}))$.
Also $\Lambda$ is zero in the first column and by (4) there can be no occurrence of $y^{1,2}$ in the first column,
so an occurrence of $y^{3,1}y^{1,2}y^{3,3}$ in $\det(A)$ must involve $y^{3,1}$ in the first column, which only occurs at position $(j,1)$.
So $y^{3,1}y^{1,2}y^{3,3}$ occurs in $\det(A)$ iff $y^{1,2}$ occurs in $\det(A(\{1,j\},\{1,n\}))$.
But by the special form of $\Lambda$ it follows that the degree 1 term of $\det(A(\{1,j\},\{1,n\}))$
is a nonzero scalar multiple of $X_{n,j}$.
With (16b), it follows that $y^{3,1}y^{1,2}y^{3,3}$ appears in $\det(A)$.
This is a contradiction to (2).
Therefore we ruled out case 2.2.1.

\subsection*{Case 2.2.2: In column 1 $y^{3,1}$ only appears in rows 2 and/or 3}
Clearly $2 \leq k \leq 3$.

If $k=2$, then $X^{1,2}_{n-1,2} \neq 0$.
By (4) and case assumption 2.2 it follows $X^{1,2}_{n,2} \neq 0$,
in contradiction to (10e).

So from now on assume that $k=3$.
In particular $X^{3,1}_{3,1}\neq 0$.
We adjust the argument from case 2.2.1 as follows.

Since $y^{3,3}$ occurs only at position $(1,n)$ and $\Lambda$ is zero in the first row,
$y^{3,1}y^{1,2}y^{3,3}$ occurs in $\det(A)$ iff $y^{3,1}y^{1,2}$ occurs in $\det(A(\{1\},\{n\}))$.
Also $\Lambda$ is zero in the first column and by (4) there can be no occurrence of $y^{1,2}$ in the first column,
so an occurrence of $y^{3,1}y^{1,2}y^{3,3}$ in $\det(A)$ must involve $y^{3,1}$ in the first column.
Since $k=3$ this occurs at position $(3,1)$, but by case assumption 2.2.2 it might also occur at position $(2,1)$.
But even though $y^{3,1}$ can appear at position $(2,1)$, this $y^{3,1}$ cannot contribute
to the coefficient of $y^{3,1}y^{1,2}y^{3,3}$ in $\det(A)$, because the special form of $\Delta$ together with (10e)
ensures that $\det(A(\{1,2\},\{1,n\}))$ has no term $y^{1,2}$.
So $y^{3,1}y^{1,2}y^{3,3}$ occurs in $\det(A)$ iff $y^{1,2}$ occurs in $\det(A(\{1,3\},\{1,n\}))$.
But by the special form of $\Lambda$ it follows that the degree 1 term of $\det(A(\{1,3\},\{1,n\}))$
is a nonzero scalar multiple of $X_{n,3}$.
Using the case assumption 2.2, it follows that $y^{3,1}y^{1,2}y^{3,3}$ appears in $\det(A)$.
This is a contradiction to (2).
Therefore we ruled out case 2.2.2.

%

\begin{thebibliography}{10}

\bibitem{2015arXiv150502205A}
J.~{Alper}, T.~{Bogart}, and M.~{Velasco}, \emph{{A lower bound for the
  determinantal complexity of a hypersurface}}, ArXiv e-prints (2015).

\bibitem{DBLP:journals/corr/AravindJ15}
N.~R. Aravind and Pushkar~S. Joglekar, \emph{On the expressive power of
  read-once determinants}, CoRR \textbf{abs/1508.06511} (2015).

\bibitem{DBLP:journals/corr/BurgisserIP16}
Peter B{\"{u}}rgisser, Christian Ikenmeyer, and Greta Panova, \emph{No
  occurrence obstructions in geometric complexity theory}, CoRR
  \textbf{abs/1604.06431} (2016).

\bibitem{MR3340547}
Melody Chan and Nathan Ilten, \emph{Fano schemes of determinants and
  permanents}, Algebra Number Theory \textbf{9} (2015), no.~3, 629--679.
  \MR{3340547}

\bibitem{MR0029360}
Jean Dieudonn{\'e}, \emph{Sur une g\'en\'eralisation du groupe orthogonal \`a
  quatre variables}, Arch. Math. \textbf{1} (1949), 282--287. \MR{0029360
  (10,586l)}

\bibitem{Gre11}
Bruno Grenet, \emph{{An Upper Bound for the Permanent versus Determinant
  Problem}}, Theory of Computing (2014), Accepted.

\bibitem{2015arXiv151203798I}
C.~{Ikenmeyer} and G.~{Panova}, \emph{{Rectangular Kronecker coefficients and
  plethysms in geometric complexity theory}}, ArXiv e-prints (2015).

\bibitem{MR2745772}
G{\'a}bor Ivanyos, Marek Karpinski, and Nitin Saxena, \emph{Deterministic
  polynomial time algorithms for matrix completion problems}, SIAM J. Comput.
  \textbf{39} (2010), no.~8, 3736--3751. \MR{2745772 (2012h:68101)}

\bibitem{MR3169697}
Harlan Kadish and J.~M. Landsberg, \emph{Padded polynomials, their cousins, and
  geometric complexity theory}, Comm. Algebra \textbf{42} (2014), no.~5,
  2171--2180. \MR{3169697}

\bibitem{LRpermdet}
J.M. Landsberg and Nicolas Ressayre, \emph{Permanent v. determinant: an
  exponential lower bound assuming symmetry and a potential path towards
  valiant's conjecture}, arXiv:1508.05788 (2015).

\bibitem{MV:97}
Meena Mahajan and V.~Vinay, \emph{Determinant: combinatorics, algorithms, and
  complexity}, Chicago J. Theoret. Comput. Sci. (1997), Article 5, 26 pp.
  (electronic). \MR{1484546 (98m:15016)}

\bibitem{MR2126826}
Thierry Mignon and Nicolas Ressayre, \emph{A quadratic bound for the
  determinant and permanent problem}, Int. Math. Res. Not. (2004), no.~79,
  4241--4253. \MR{MR2126826 (2006b:15015)}

\bibitem{MS1}
Ketan~D. Mulmuley and Milind Sohoni, \emph{Geometric complexity theory. {I}.
  {A}n approach to the {P} vs.\ {NP} and related problems}, SIAM J. Comput.
  \textbf{31} (2001), no.~2, 496--526 (electronic). \MR{MR1861288
  (2003a:68047)}

\bibitem{MS2}
\bysame, \emph{Geometric complexity theory. {II}. {T}owards explicit
  obstructions for embeddings among class varieties}, SIAM J. Comput.
  \textbf{38} (2008), no.~3, 1175--1206. \MR{MR2421083}

\bibitem{Nisan:1991:LBN:103418.103462}
Noam Nisan, \emph{Lower bounds for non-commutative computation}, Proceedings of
  the Twenty-third Annual ACM Symposium on Theory of Computing (New York, NY,
  USA), STOC '91, ACM, 1991, pp.~410--418.

\bibitem{MR2535881}
Ran Raz, \emph{Multi-linear formulas for permanent and determinant are of
  super-polynomial size}, J. ACM \textbf{56} (2009), no.~2, Art. 8, 17.
  \MR{2535881 (2011a:68043)}

\bibitem{vali:79-3}
Leslie~G. Valiant, \emph{Completeness classes in algebra}, Proc.~11th ACM STOC,
  1979, pp.~249--261.

\bibitem{MR910987}
Joachim von~zur Gathen, \emph{Permanent and determinant}, Linear Algebra Appl.
  \textbf{96} (1987), 87--100. \MR{MR910987 (89a:15005)}

\end{thebibliography}

\def\cdprime{$''$} \def\cprime{$'$} \def\cprime{$'$} \def\cprime{$'$}
  \def\Dbar{\leavevmode\lower.6ex\hbox to 0pt{\hskip-.23ex \accent"16\hss}D}
  \def\cprime{$'$} \def\cprime{$'$} \def\cdprime{$''$} \def\cprime{$'$}
  \def\cprime{$'$} \def\Dbar{\leavevmode\lower.6ex\hbox to 0pt{\hskip-.23ex
  \accent"16\hss}D} \def\cprime{$'$} \def\cprime{$'$} \def\cprime{$'$}
  \def\cprime{$'$} \def\Dbar{\leavevmode\lower.6ex\hbox to 0pt{\hskip-.23ex
  \accent"16\hss}D} \def\cprime{$'$} \def\cprime{$'$}
\providecommand{\bysame}{\leavevmode\hbox to3em{\hrulefill}\thinspace}
\providecommand{\MR}{\relax\ifhmode\unskip\space\fi MR }
\providecommand{\MRhref}[2]{%
  \href{http://www.ams.org/mathscinet-getitem?mr=#1}{#2}
}
\providecommand{\href}[2]{#2}

\end{document}